\documentclass[12pt, reqno]{amsart}
\usepackage[utf8]{inputenc}
\usepackage{amsmath}
\usepackage{amssymb,amscd,amsfonts,amssymb, color, xcolor, enumerate,latexsym}
\usepackage{amsthm}
\usepackage[pdftex,bookmarks,colorlinks]{hyperref}
\hypersetup{colorlinks=false,
unicode=true,pdftitle={Interpolation by entanglement breaking channels, random unitary channels },pdfauthor={Arnab Roy and Saikat Patra},bookmarks=false} 
\numberwithin{equation}{section}
\newtheorem{theorem}{Theorem}[section]

\newtheorem{lemma}[theorem]{Lemma}
\theoremstyle{definition}
\newtheorem{defn}{Definition}[section]

\newtheorem{prop}{Proposition}
\newcommand{\tr}{\mathrm{Tr\,}}

\theoremstyle{remark}

\title[Interpolation of Quantum Channels Using Conic Programs]{Interpolation by Different Types of Quantum Channels Using Conic Programs }
\author{Arnab Roy}
\address{Arnab Roy, Address: Department of Mathematical Sciences, University of Delaware, Ewing Hall, Newark, Delaware, 19716}
\email[Arnab Roy]{\href{arnabroy@udel.edu}{arnabroy@udel.edu}}
\author{Saikat Patra}
\address{Saikat Patra, Address: Department of Mathematical Science, Indian Institute of
Science, Education, \& Research (IISER), Berhampur, Transit campus -
Govt. ITI, NH 59, Berhampur 760 010, Ganjam, Odisha, India}
\email[Saikat Patra]{\href{saikatp@iiserbpr.ac.in}{saikatp@iiserbpr.ac.in}}

\begin{document}
\begin{abstract}
We have found conic programs for getting different types of quantum channels as outputs of interpolation problems.
Afterwards, we have generalised our results for getting channels that belong to a convex set as outputs of the interpolation problem. Then we show the existence of an Entanglement breaking channel for orthogonal sets of input and output matrices.
\end{abstract}
\maketitle
\tableofcontents

\newpage

\section{Introduction}
\par One of the standard problems in classical information theory is the
following: given two probability distributions, find a stochastic method to distinguish them. This can be answered probabilistically and a
relevant area of research is known as hypothesis testing. A generalization of this
problem with more than two probability distributions is known as multiple
hypothesis testing. The quantum version of these problems, i.e. quantum
hypothesis testing also has a rich history. One can see the book of Hayashi
\cite{hayashi-1} and the references therein to know the history of the problem.

\par In this paper, we consider a quantum version of the problem from a matrix
analysis point of view. Following the notion of Kolmogorov, a probability space
can be represented by a triplet $(\Omega, \Sigma, f)$, where $\Omega$ is the
sample space, $\Sigma$ is the set of events, and $f$ is a probability
distribution. A quantum probability space can similarly be represented by a
triplet $(H, \mathcal{P}(H), \rho)$, where sample space is represented by a
complex separable Hilbert space $H$, set of events is represented by the set of
projections on $H$ denoted by $\mathcal{P}(H)$, and a probability distribution is
replaced by a positive semi-definite trace class Hermitian operator $\rho$ where
$\tr(\rho)=1$. This is called a state. If $\rho$ is a rank one projection, it is
also called a pre-state. See the book of Meyer \cite{meyer} or Parthasarathy  \cite{krp6}
for more details regarding this approach. in this paper, we only consider Hilbert
spaces of finite dimensions. As a result, the operators involved are all matrices.
Physical transformations between two quantum systems are given by completely
positive trace-preserving transformations called quantum channels.
A quantum channel $\Phi: \mathcal{B}(H) \to  \mathcal{B}(H)$ can be represented
by a finite set of operators $\{V_1,V_2, \cdots, V_k\}$ acting on $H$ such that
for all $X \in \mathcal{B}(H)$
\begin{equation}\label{eq1.1}
	\Phi(X) = \sum_{j=1}^k V_j X V_j^\dag, \qquad  \text{ where }
	\sum_{j=1}^k V_j^\dag V_j=I.
\end{equation}
Here $I$ denotes the identity operator on $H$. The operators $V_j$ are called the
Kraus operators. (See chapter 3 from the book of  Bhatia \cite{bhatia1} for a  proof.
This representation is not unique.
However, a minimal representation with respect to the number of terms can be achieved.

\par The problem can be stated as follows: Let $A_1,\cdots, A_n$ be a set of
input states and $B_1, \cdots, B_n$ be a set of output states, both on
$\mathcal{B}(H)$. The task is to construct (if possible) a quantum channel $\Phi:
\mathcal{B}(H) \to \mathcal{B}(H)$ such that $\Phi(A_i) =B_i$ for each $i=1,\, 2,
\cdots, n.$ This can be considered as a version of the interpolation problem in
operator theory where an interpolating function is a completely positive trace-preserving map.  When $\dim H=2$ and $n=2$ the problem has been completely solved
by  Alberti and Uhlmann \cite{alberti-uhlmann-80}, with some more developments in
\cite{MR649773}. Further developments, in the sense of hypothesis testing, were done
by Jen\v{c}ov\'{a} \cite{MR2997404},\cite{MR2661520}. However, the most general form of the above problem remains elusive. The classical version of the problem was solved by Blackwell
\cite{MR56251}. For further developments in this direction see the paper of Ruch
et al \cite{MR586660} or the book of Torgersen \cite{MR1104437}. All of these
processes use various properties of matrix majorization.

\par The problem statement can be explained in the language of operator
systems. Consider the operator system $\mathcal{S}_A$ spanned by $A_i$'s and the
map $\phi:\mathcal{S}_A \to \mathcal{B}(H)$ such that $\phi(A_i)=B_i$. By
definition, this is a completely positive map. By Arveson's extension theorem,
this can be extended to a completely positive map $\Phi: \mathcal{B}(H) \to
\mathcal{B}(H)$ \cite{paulsen}. This extension is not unique and it is not clear
whether this extension is in a quantum channel or not. 

\par Li and Poon approached this problem from a matrix analysis point of
view \cite{MR2837768}.
They assumed that the input and output matrices are also commuting. As a result,
the problem is reduced to a problem on diagonal matrices, which can be approached by
standard methods.   This had been further
extended by Hsu et al \cite{hsu-jfa-14} where they considered only one input $A$
and only one output $B$ where the operators are either compact or of Schatten
$p$-class. Following techniques as in \cite{MR2837768}, they showed that the
result also holds for a class (compact or Schatten $p$) of input operators
$A_i$'s and corresponding output operators $B_i$'s where they are mutually
commuting.

\par From a quantum information perspective Heinosaari et al \cite{heinosaari-12}
applied semi-definite programming techniques to construct such channels. They also
asked a sequential version of the problem, i.e. to say whether there exists a
sequence of quantum channels $\Phi_j$ such that $\lim_{j\to \infty} \|\Phi_j(A_i)
-B_i\|=0$ for all $i=1, \cdots, n$. They observed that there exist cases when the sequential version of the problem has solutions whereas there does not exist
any such channel $\Phi$. Finally, Gour et al  \cite{Gour2018} showed several
equivalent statements of the problem connecting it with quantum thermodynamics.
Following the classical problem, which relies heavily on majorization, they
consider the problem as a non-commutative version of majorization. 

\par In \cite{MR3671511} Davidson et al approached interpolation problems from an
operator theoretic perspective. They did these works in the context of $d$-tuples
of operators on Complex Hilbert space and they used duality of the Hilbert space
most of their work. Basically, they associate with d-tuples of bounded operators
to the matrix range which is introduced by Arveson.
That is, a d-tuple of operators
can be mapped onto another d-tuple of operators by a UCP map. This result
connects the problem with the joint numerical range which was also initially noticed
in  \cite{hsu-jfa-14} in a very restricted class of operators. In  \cite{FB2891869} shows that for given state space input $\mathcal{A}$ and output $\mathcal{B}$ on $\mathcal{H}_1$ and $\mathcal{H}_2$ respectively with  $\mathcal{B}$ is commutative. If there exist a linear map $\Psi : \mathcal{L}( \mathcal{H}_1) \to \mathcal{L}( \mathcal{H}_2) $ such that $\Psi (A) \in \mathcal{B} $ for all $ A \in \mathcal{A}$ and dual  $\Psi^* $ map test output state to input state. Then there exist a completely positive trace preserving map $ \Phi : \mathcal{L}( \mathcal{H}_1) \to \mathcal{L}( \mathcal{H}_2) $ such that $\Psi (A) = \Phi (A)$ far all $ A \in \mathcal{A}$.

\par A quantum channel $\Phi$ is said to be
entanglement breaking if for any state $\rho$ acting on $H \otimes H$ the output
$(1 \otimes \Phi) (\rho)$ is a separable state. This has been introduced by
Horodecki et al \cite{MR2001114}. It can be shown that for such a map $\Phi$
there is a Kraus representation (which need not be optimal) such that each Kraus
operator is of rank 1. In this paper, we tried to address the following questions:
\begin{enumerate}
 \item First we have reconstructed the SDPs in \cite{heinosaari-12} so that
	 we can identify the EBT maps.
	
 \item	 Secondly we have solved the interpolation problem with entanglement breaking
	 map when initial states are orthogonal.
 \end{enumerate}
 We notice that Li and Du \cite{MR3336734} also addressed a similar problem.
 But in their case, there was one input and one output operator, and the separable
 Hilbert spaces were of any arbitrary dimensions. This paper considers the
 problem only on finite-dimensional Hilbert spaces.


\section{Preliminary Concepts and Results}
In this section, we are going to list the preexisting results we are going to use in our discussion.\\

 Let $ \mathbb{M}_n$ be the set of all complex matrices and $ \mathbf{H}_n$ be the set of all Hermitian matrices in $ \mathbb{M}_n$. A matrix $A\in \mathbb{M}_n$ is said to be positive if   $ A \in \mathbf{H}_n$ and all of its eigenvalues are nonnegative. A linear map $ \Phi : \mathbb{M}_n \rightarrow  \mathbb{M}_m$ is positive if it maps positive matrices to positive matrices. Let, $\mathbb{M}_k(\mathbb{M}_m)$ be the space of all $k\times k$  block matrices $[[A_{ij}]]$ where $A_{ij} \in \mathbb{M}_m$. A linear map $ \Phi : \mathbb{M}_n \rightarrow  \mathbb{M}_m$ is completely positive if for each k, the map $  I  \otimes \Phi :\mathbb{M}_k(\mathbb{M}_n) \rightarrow \mathbb{M}_k(\mathbb{M}_m)$ defined by $ (I_k \otimes \Phi)(A_{ij})=[\Phi(A_{ij})] $ is positive.\\
 
 In the Introduction, we said the interpolation problem is a version of operator theory.
 Now we want to see for a given set of input states $\{A_1,A_2,\dots,A_n\}$ and output states $\{B_1,B_2,\dots,B_n\}$,
 if there exists a condition for the existence of a completely positive trace-preserving map $ \Phi: \mathbb{M}_n \rightarrow  \mathbb{M}_m$ such that $\Phi(A_i)=B_i$ for $i=1,2,\dots,n$.
 In particular, for $i=1$, that is for single matrices A and B, can we find a condition for the existence of completely positive trace-preserving maps such that $ \Phi(A)=B$?
 The following results by Li and Poon show that condition for the existence of a completely positive trace-preserving map for Hermitian matrices.\\

\begin{theorem}\label{th2.1}\cite{MR2837768}
Suppose  $A\in\mathbf{H}_n$ and $B\in\mathbf{H}_m$ have eigenvalues $a_1,\dots,a_n$ and $b_1,\dots,b_m$ respectively.
Let $ a=(a_1,\dots,a_n)$ and $b=(b_1\dots ,b_m)$.
The following conditions are equivalent:
\begin{enumerate}
 \item There is a completely positive map trace-preserving  $\Phi:\mathbb{M}_n\rightarrow\mathbb{M}_m$ such that $\Phi(A)=B$
 \item There exists an $n\times m$ row stochastic matrix $\mathbf{D}$ such that $b=a\mathbf{D}$.
\item We have $tr(A)=tr(B)$ and $\sum_{p=1}^n|a_p|\geq\sum_{q=1}^m|b_q|$.
 \end{enumerate}
 \end{theorem}
 
 After adding some condition on the matrices we can find the condition of the existence of a completely positive trace-preserving map for finite input and output matrices. More precisely, If we consider inputs $\{A_1,A_2,\dots,A_n\}$ and outputs $\{B_1,B_2,\dots,B_n\}$ are commuting family of Hermitian matrices. Then the existence of a completely positive trace-preserving for the above family of matrices such that $\Phi(A_i)=B_i$ for $i=1,2,\dots,n$ is as follows.
 
 \begin{theorem}\label{th2.2}\cite{MR2837768}
Let $\{A_1,\dots ,A_k\}$ in $\mathbf{H}_n$ and $\{B_1,\dots , B_k\}$ in $\mathbf{H}_m$ be two commuting families.
Then there exist  unitary matrices $U\in\mathbb{M}_n ~and~V\in\mathbb{M}_m$ such that $U^*A_iU$ and $VB_iV^*$ are diagonal matrix with diagonals $a_i=(a_{i1},\dots ,a_{in})$ and $b_i=(b_{i1},\dots ,b_{im})$ respectively, for i=1,...k. The following conditions are equivalent:

\begin{enumerate}
\item There is a completely positive map $\Phi : \mathbb{M}_n \rightarrow \mathbb{M}_m$ such that $\Phi  (A_i)=B_i $ for $i=1,\dots ,k$.
\item There is an $n\times m$ non-negative matrix $D=\begin{pmatrix}
d_{pq}
\end{pmatrix}_{pq}$ such that $(b_{ij})=(a_{ij})D$.
\end{enumerate}

\end{theorem}

 Now from the definition of entanglement breaking, we know that any output state of an entanglement breaking channel is separable.
 From \cite{H1}, we also know how to check the separability of a state.
 The following result shows how to characterize entanglement-breaking maps. 

\begin{theorem}\label{th2.3}
Let $\phi$ be a completely positive map on $\mathcal{B}(H_1)\otimes\mathcal{B}(H_2)$ where $H_1=H_2=\mathbb{C}^n$.
Then the following are equivalent:
\begin{enumerate}
    \item $\phi$ can be written as $\phi(\rho)=\sum_k R_ktr[F_k\rho]$\footnote{This form of $\phi$ is called Holevo form, see \cite{Holevo-coding}} with $F_k$ positive semi-definite for every $k$.
    \item $\phi$ is entanglement breaking.
    \item $ I \otimes\phi(|\beta\rangle\langle\beta|)$ is separable for any maximally entangled state $\beta=n^{-\frac{1}{2}}\sum_{j=1}^n |j\rangle\otimes|j\rangle$.
    \item $\phi$ can be written in operator sum form using only Kraus operators $\{A_l\}_l$ of rank one.
    \item $\gamma\circ\phi$ is completely positive for all positivity preserving maps $\gamma$.
    \item $\phi\circ\gamma$ is completely positive for all positivity preserving maps $\gamma$.
\end{enumerate}
A corresponding equivalence holds for CPT and EBT maps with the additional
conditions that $\sum_k F_k= I $, the Kraus operators $A_l$ satisfy $\sum_l A_lA_l^*= I $,
and $\gamma$ is trace-preserving.
\end{theorem}

Using these theorems we have found a completely positive trace-preserving map between two families of positive matrices where initial matrices are orthogonal and we have also found whether the map is entanglement breaking or not.

The first problem is inspired by the existence of a completely positive trace-preserving extension formulated as a semi-definite program in \cite{heinosaari-12}. 
 \\

Let $\{X_i\}_{i=1}^N$ and $\{Y_i\}_{i=1}^N$ be Hermitian matrices.
If linear map $\tilde{T}:\mathbb{M}_d\rightarrow \mathbb{M}_{d'}$ is the approximation of $T:\mathbb{M}_d\rightarrow \mathbb{M}_{d'}$ such that $T(X_i)=Y_i$. Define this $\Delta$ to be a measure of the  functional\\
 \begin{equation}\label{eq2.1}
\Delta(\tilde{T})=\sum_{i=1}^N\Vert\tilde{T}(X_i)-Y_i\Vert_1 
\end{equation}
which calculate a linear map $\tilde{T}$ is how much close to $ T$.\\

So,
\begin{equation}\label{eq2.2}
\Delta=\inf\{\Delta(\tilde{T})|\tilde{T}:\mathbb{M}_d\rightarrow\mathbb{M}_{d'}~ ~is~completely~~positive\} 
\end{equation}
calculates how well the linear map $T:\mathbb{M}_d\rightarrow \mathbb{M}_{d'}$ such that $T(X_i)=Y_i$ can be approximated by a completely positive map defined on $\mathbb{M}_d$.
 Heinosaari et al \cite{heinosaari-12} have identified all this as the following SDP problem.
 
\begin{theorem}\label{th2.4}
Let, $X_i\in\mathbb{M}_d$ and $Y_i\in\mathbb{M}_{d'}$ be Hermitian matrices.
The optimization \ref{eq2.1}-\ref{eq2.2} whose optimal value $\Delta$ quantifies how well the map $X_i\rightarrow Y_i$  can be approximated with a completely positive map on $\mathbb{M}_d$ can be formulated as an SDP.
Up to the negative sign, it is the dual of the following SDP\\
\begin{equation}\label{2.3}
\Gamma=\inf\sum_{i=1}^Ntr[Y_i^TH_i]
\end{equation}
Subject to constraints,
\begin{equation}\label{2.4}
\sum_{i=1}^NX_i\otimes H_i\geq0\\
~~and ~~\Vert H_i \Vert \leq 0 ~~\forall i
\end{equation}
where $H_i$'s are Hermitian matrices.
The optimal value $\Gamma$ of the SDP \ref{2.3} - \ref{2.4} is equal to $0$ if and only if the map $X_i\rightarrow Y_i $ is completely positive on $\mathcal{S}=span\{X_i\}$; The optimal values both of the SDP are related by $ \Delta = -\Gamma$. 
\end{theorem}
In the paper, \cite{heinosaari-12} the existence of a completely positive trace-preserving extension is formulated as the following two semi-definite programs.\\

The first one is:

\begin{equation}\label{eq2.5}
\Delta_1=\inf\{\omega\lambda+\sum_i tr[P_i+Q_i]\} 
\end{equation}
Subject to
\begin{equation}\label{eq2.6}
C,P_i,Q_i\geq0 
\end{equation}
\begin{equation}\label{eq2.7}
P_i-Q_i=tr_2[C( I \otimes X_i^T)]-Y_i~~\forall i 
\end{equation}
\begin{equation}\label{eq2.8}
-\lambda I \leq tr_1[C]- I \leq\lambda I 
\end{equation}

The second one is:
\begin{equation}\label{eq2.9}
\Delta_2=\inf\{\sum_i tr[P_i+Q_i]\} 
\end{equation}
Subject to
\begin{equation}\label{eq2.10}
C,P_i,Q_i\geq 0 
\end{equation}
\begin{equation}\label{eq2.11}
P_i-Q_i = tr_2[C( I \otimes X_i^T)]-Y_i~~\forall i 
\end{equation}
\begin{equation}\label{eq2.12}
tr_1[C] =  I 
\end{equation}
 
We have modified these SDPs in the subsequent sections to identify EBT maps.

\section{An Optimization for Identifying EBT Maps}
\subsection{A Finite Sequence of Abstract Convex Optimizations}

In this section, our goal is to reconstruct the SDP's \ref{eq2.5}-\ref{eq2.8} and \ref{eq2.9}-\ref{eq2.12} so that we can identify the EBT maps.\\

We know from theorem \ref{th2.3} a map $T$ is EBT if and only if its corresponding Choi matrix $C$ is separable.
Using theorem 1, of \cite{H1} we can identify all separable states by an abstract convex optimization problem:

\begin{equation}\label{eq4.1}
\Lambda_C=\inf\{ tr[AC]\}
\end{equation}
 subject to
\begin{equation}\label{eq4.2}
tr[AR\otimes S]\geq 0 
\end{equation}
  for all projection matrices $R$ and $S$.\\

Optimization \ref{eq4.1}-\ref{eq4.2} is a convex optimization because the objective function $f_0(A)=tr[AC]$ and the constraint functions $f_{R,S}(A)=tr[AR\otimes S]$ are linear functions over a convex set.\\

Then again the constraint \ref{eq4.2} in the optimization \ref{eq4.1}-\ref{eq4.2} can be reformulated as the following abstract convex optimization problem:

\begin{equation}\label{eq4.3}
\Gamma_A=\inf\{ tr[AR\otimes S]\}
\end{equation}
subject to
\begin{equation}\label{4.4}
R^2 = R
\end{equation}

\begin{equation}\label{eq4.5}
S^2 = S
\end{equation}\\
Hence the optimization can be summarised as the following:\\

First, optimize $\Lambda_C$ over all $A$ for which $\Gamma_A\geq0$.
Then optimize $\Delta_i$ over all $C$ for which $\Lambda_C\geq0$.

\subsection{a cone program similar to an SDP}
In this section, we are going to use the definition of separability to construct a cone program that is almost an SDP(i.e just like an SDP, for this optimization the constraint functions run over a smaller convex cone of positive semi-definite matrices) that will identify the EBT maps interpolating between Hermitian $\{ X_i\}_i$ and $\{ Y_i\}_i$.\\

Before that, we have to note that in the SDP's \ref{eq2.5}-\ref{eq2.8} and \ref{eq2.9}-\ref{eq2.12}, $C$ is running over the convex cone of all positive semi-definite matrices.
On the other hand, this $C$ turns out to be the Choi matrix of some completely positive map.
For a quantum channel $\phi$, the corresponding Choi matrix $C_\phi$ must be a state in $\mathbb{M}_{n^2}$.
In spite of this fact, in the SDP's \ref{eq2.5}-\ref{eq2.8} and \ref{eq2.9}-\ref{eq2.12} $C$ were taken over all positive semi-definite matrices.\\

This can only mean that, if $\Delta_i=0$ then the minimum value of $w\lambda+\sum_{i}tr[P_i+Q_i]$ must be obtained for some $C$ with trace $1$.\\

\begin{lemma}\label{lm4.1}
The property "\textbf{if $\Delta_i=0$ then the minimum value of $w\lambda+\sum_{i}tr[P_i+Q_i]$ must be obtained for some $C$ with trace $1$}"will hold if we define an optimization similar to \ref{eq2.5}-\ref{eq2.8} where $C$ is in a subset $\mathcal{S}$ of positive semi-definite matrices instead of the full set and the same holds for \ref{eq2.9}-\ref{eq2.12}.
\end{lemma}
\begin{proof}
The proof follows from the fact that the minimum of a function over a subset is always greater than or equal to the minimum over the full set.\\

Let $\Delta_S=inf\{ w\lambda+\sum_i tr[P_i+Q_i]\}$ subject to $C\in\mathcal{S}$.
Then $\Delta_1\leq\Delta_S$.\\

Now, let us assume the claim in the lemma is not true. Then there is some $C\in\mathcal{S}$ with $tr[C]\neq 1$ and $\Delta_S=0$ for $C$.\\

That means $0\leq\Delta_1\leq\Delta_S=0\Rightarrow\Delta_1=0$ for $C$.
Hence the optimization identifies a completely positive map which is not a quantum channel.
This is a contradiction.\\

The proof for SDP \ref{eq2.9}-\ref{eq2.12} is similar.
\end{proof}

Keeping this property in mind let us define $\mathcal{C}=\{\sum_{i=1}^n p_iA_i\otimes B_i:p_i\geq 0,A_i$ and $B_i$ are states in $\mathbb{M}_n\forall i,\}$.
\begin{theorem}\label{th4.2}
$\mathcal{C}$ is a convex cone of positive semi-definite matrices.
\end{theorem}

\begin{proof}
Let's choose $C_1,C_2\in\mathcal{C}$ and $\lambda\in(0,1)$ arbitrarily.
Let, $C_j=\sum_{i=1}^n p_{ji}A_{ji}\otimes B_{ji}$ for $i=1,2$.[Even if $C_1$ and $C_2$ have different number of components, we can adjust this by adding sufficiently many zeroes to the element with fewer components.]\\

Then, $\lambda C_1+(1-\lambda)C_2=\sum_{i=1}^n\lambda p_{1i}A_{1i}\otimes B_{1i}+\sum_{k=1}^n(1-\lambda)p_{2k}A_{2k}\otimes B_{2k}\in\mathcal{C}$.\\

Hence $\mathcal{C}$ is convex.\\

Let $r\in\mathbb{R}^+$ be chosen arbitrarily.
Clearly $rC\in\mathcal{C}$ for any $C\in\mathcal{C}$.\\

Hence $\mathcal{C}$ is a cone.
\end{proof}

Now we are in a position to construct our cone program.
We define the Optimization $\mathcal{O}$ as the following:

\begin{equation}\label{eq4.6}
\Delta =\inf\{ w\lambda+\sum_i tr[P_i+Q_i]\},
\end{equation}

subject to,
\begin{equation}\label{eq4.7}
P_i,Q_i\geq 0,
\end{equation}
\begin{equation}\label{eq4.8}
P_i-Q_i =tr_2[C( I \otimes X_i^T)]-Y_i\forall i,
\end{equation}
\begin{equation}\label{eq4.9}
C\in\mathcal{C},
\end{equation}
\begin{equation}\label{eq4.10}
    -\lambda I \leq tr_1[C]- I \leq\lambda I 
\end{equation}\\

There exists an interpolating EBT if and only if $\Delta=0$.

\section{An optimization problem to identify Interpolating degradable channels}

A Channel $\Phi : \mathcal{S}(\mathcal{H}_1) \rightarrow \mathcal{S}(\mathcal{H}_2) $ is said to be degradable if there exists a Completely positive trace preserving map  $\Psi :\mathcal{S}(\mathcal{H}_2) \rightarrow \mathcal{S}(\mathcal{E}) $ such that $\Phi^C = \Psi\circ \Phi $ . Where $\Phi^C$ is a complementary channel and $\mathcal{S}(\mathcal{H}) $ = State space on $\mathcal{H}$, E is the environment.

\par A channel is  $\Phi : \mathcal{S}(\mathcal{H}_1) \rightarrow \mathcal{S}(\mathcal{H}_2) $ is said to be $\epsilon $ - degradable if there exists a channel $\Psi :\mathcal{S}(\mathcal{H}_2) \rightarrow \mathcal{S}(\mathcal{E}) $ such that 
\begin{equation}
||\Phi^C - \Psi\circ \Phi ||_\diamond \leq \epsilon
\end{equation}
Where $|| \Phi ||_\diamond = || \Phi \otimes \mathcal{I}_{\mathcal{H}_1} ||_1 $, $ ||A||_1=tr\sqrt{A^*A}$

Watrous \cite{JhonW2009} proved that for two-channel $\Phi_1 , \Phi_2 : \mathcal{S}(\mathcal{H}_1) \rightarrow \mathcal{S}(\mathcal{H}_2) $ the diamond norm of their difference  $||\Phi_1 - \Phi_2 ||_\diamond $ can be expressed as a semi-definite programming 
\begin{equation}
||\Phi_1 - \Phi_2 ||_\diamond = 2 \inf_{Z} || tr_{2}(Z)||_\infty
\end{equation}
subject to
\begin{equation}
Z\geq C_{\Phi_{1}}- C_{\Phi_{2}}
\end{equation}
\begin{equation}
Z\geq 0
\end{equation}

Where $C_{\Phi_{1}- \Phi_{2}}$ is the Choi-Jamiolokowski representation of $\Phi_{1}- \Phi_{2}$.

\par From the paper \cite{David2017}  Approximate degradable Quantum channel
\begin{eqnarray*}
\epsilon_{\Phi}=\inf_{\Psi}||\Phi^C - \Psi\circ \Phi||_{\diamond}
\end{eqnarray*}
Subject to $\Psi :\mathcal{S}(\mathcal{H}_2) \rightarrow \mathcal{S}(\mathcal{E}) $ is Completely positive trace preserving.
\par So $\epsilon_{\Phi}$ can be written as
\begin{equation}
\epsilon_{\Phi}=2\inf_{\Psi}\inf_{Z}||tr_{E}(Z)||_\infty
\end{equation} 
Subject to,
\begin{gather*}
    Z\geq C_{\Phi^{C}}- C_{\Psi\circ \Phi}\\
Z\geq 0\\
C_{\Psi}\geq 0\\
tr_{E}(C_\Psi)=\mathcal{I}_{\mathcal{H}_2}
\end{gather*}

\par Now we can construct the above semi-definite programming so that the programming can identify Interpolating degradable maps.
\par If $X_i$ and $Y_i$ are input and output states then to find an Interpolating map that is degradable given by the following semi-definite programming,
\begin{equation}
\epsilon_{\Phi}=2\inf_{C_\Phi,Z}||tr_{E}(Z)||_\infty
\end{equation}
subject to,
\begin{gather*}
    Z \geq C_{\Phi^C} - C_{\Psi\circ \Phi}\\
    Z\geq 0\\
    C_{\Psi} \geq 0\\
    tr_{E}(C_\Psi)=tr_{\mathcal{H}_1}(C_\Phi)=\mathcal{I}_{\mathcal{H}_2}\\
tr_{\mathcal{H}_2}(C_\Phi(\mathcal{I}\otimes {X_i}^T))=Y_i \quad \forall \;i = 1,2\cdots ,n
\end{gather*}

\section{An optimization problem to identify random unitary channels}
This section will construct another optimization problem similar to \ref{eq4.7}-\ref{eq4.10} to identify random unitary channels instead of EBT channels.\\

Our first step will be to characterise the Choi matrix of a random unitary channel.
Let $T$ be a random unitary channel.
Then $T(A)=\sum_{i=1}^n p_iU_i^*AU_i$ where $U_i$ is unitary and $p_i>0$ for each $i$ and $\sum_{i=1}^n p_i=1$.\\

Then the Choi matrix of $T$ is:
\begin{eqnarray*}\label{eq5.1}
    \begin{split}
        C_T =\begin{pmatrix}
\begin{pmatrix}
T(E_{ij})
\end{pmatrix}
\end{pmatrix}_{ij}\\
=\begin{pmatrix}
\begin{pmatrix}
\sum_{k=1}^np_kU^*_kE_{ij}U_k
\end{pmatrix}
\end{pmatrix}_{ij}\\
=\sum_{k=1}^np_k\begin{pmatrix}
\begin{pmatrix}
U^*_kE_{ij}U_k
\end{pmatrix}
\end{pmatrix}_{ij}
    \end{split}
\end{eqnarray*}\\

Hence from above equation , we can conclude that the set of Choi matrices of random unitary channels are $\mathbb{C}_U=\{\sum_{k=1}^np_k\begin{pmatrix}
\begin{pmatrix}
U^*_kE_{ij}U_k
\end{pmatrix}
\end{pmatrix}_{ij}:U_k$ is unitary matrix in $\mathbb{M}_n$, $p_k>0$ for all $k$ and $\sum_{k=1}^np_k=1\} $.\\

As explained in section \S3.2, if the $\Delta_1$ in optimization \ref{eq2.5}-\ref{eq2.8} is zero, then the minimum value must have been obtained for some state $C$ and this is also true when $C$ is taken from a subset of positive semi-definite matrices instead of the full set.\\

Hence let us define, $\mathcal{C}_U=\{\sum_{k=1}^np_k\begin{pmatrix}
\begin{pmatrix}
U^*_kE_{ij}U_k
\end{pmatrix}
\end{pmatrix}_{ij}:U_k$ is unitary matrix in $\mathbb{M}_n$, $p_k>0$ for all $k\}$\\

Now the construction of the optimization problem will follow similar steps as in section \S3.2:

\begin{equation}\label{eq5.2}
\Delta  =\inf\{ w\lambda+\sum_i tr[P_i+Q_i]\},
\end{equation}

Subject to,
\begin{gather*}
    P_i,Q_i\geq 0,\\
     P_i-Q_i =tr_2[C( I \otimes X_i^T)]-Y_i \quad \forall \; i,\\
      C\in\mathcal{C}_U\\
      -\lambda I \leq tr_1[C]- I \leq\lambda I 
\end{gather*}

\section{A generalization of The Previous Discussions}
In our previous discussions, we reconstructed the problem of interpolating EBT channels between two sets of positive matrices as convex optimization problems.
In subsection \S3.2 we were even able to construct a cone program that is almost an SDP in the sense that the constraint functions run over a smaller convex cone of positive matrices rather than over the entire set of positive matrices.
In section \S5 we did a similar treatment for random unitary channels.
Let us denote this type of optimization as semi-SDP.

\begin{defn}
A cone program $\mathcal{K}$ is a semi-SDP if its constraint functions run over a convex cone of positive matrices.
\end{defn}

Now we will see how this methodology can be used in a more general case.
By general case what we mean is that instead of EBT or unitary channels, we are going to impose an arbitrary property $\mathcal{P}$ on the set of quantum channels on $\mathbb{M}_n$ and try to check if there exists a quantum channel satisfying $\mathcal{P}$ and interpolating between two given sets of positive matrices.\\

Let us denote, $\mathbb{P}=\{ T:T$ is a quantum channel satisfying $\mathcal{P}\}$ as the corresponding set of $\mathcal{P}$.
We say $\mathcal{P}$ is convex if and only if $\mathbb{P}$ is convex.
We will show that for any convex $\mathcal{P}$ the problem mentioned before can be represented as a semi-SDP.
By the one-on-one correspondence between the set of all completely positive maps and their corresponding Choi matrices $\mathbb{P}$ is convex if and only if the set of its corresponding Choi matrices is convex.
Therefore in our further discussion, we shall use the same notation $\mathbb{P}$
to denote the set of all quantum channels satisfying $\mathcal{P}$ as well as the set of Choi matrices corresponding to the former set.
\begin{defn}
Let $\mathbb{P}$ be the corresponding set of some property $\mathcal{P}$.
We define, 
\begin{eqnarray*}
    \mathbf{K}(\mathbb{P})=\left \{\sum_{i=1}^n p_iP_i: p_i\geq0 \text{ and } P_i\in\mathbb{P} \quad\forall \; i \right\}
\end{eqnarray*}

\end{defn}

\begin{theorem}\label{th6.2}
For any $\mathbb{P}$, $\mathbf{K}(\mathbb{P})$ is a convex cone.
\end{theorem}

\begin{proof}
The proof is obvious.
\end{proof}

From this theorem we can say that $\mathbf{K}(\mathbb{P})$ is the convex cone generated by $\mathbb{P}$.\\

Here we have to note that if we formulate a semi-SDP similar to \ref{eq4.6}-\ref{eq4.10} by running $C$ over $\mathbf{K}(\mathbb{P})$ it will not necessarily be true that the optimization may take $0$ value only if $C\in\mathbb{P}$.
If this happens then this cone program fails to identify interpolating quantum channels in $\mathbb{P}$.
But we shall see in the following theorem that if we take $\mathbb{P}$ to be convex then $\Delta=0$ only if $C\in\mathbb{P}$.\\

\begin{theorem}
Any arbitrary state in $\mathbf{K}(\mathbb{P})$ is in $\mathbb{P}$ if and only if $\mathbb{P}$ is convex.
\end{theorem}

\begin{proof}
First, we shall prove the set of all states in $\mathbf{K}(\mathbb{P})$ is 
$$\mathbf{S}(\mathbb{P})= \left\{\sum_{i=1}^n p_iP_i: p_i\geq0,\sum_{i=1}^np_i=1 \text{ and } P_i\in\mathbb{P} \quad \forall \; i \right\}$$.\\

Let $C=\sum_{i=1}^n p_iP_i\in\mathbf{S}(\mathbb{P})$.
Then,
\begin{eqnarray*}
    tr[C]=tr[\sum_{i=1}^n p_iP_i]=\sum_{i=1}^n p_itr[P_i]=\sum_{i=1}^n p_i=1
\end{eqnarray*}

Hence $C$ is a state in $\mathbf{K}(\mathbb{P})$.\\

Let $D=\sum_{i=1}^n p_iP_i$ be an arbitrary state in  $\mathbf{K}(\mathbb{P})$.
Then $p_i\geq 0\forall i$ and $tr[D]=1$.\\
\begin{eqnarray*}
    1=tr[D]=tr[\sum_{i=1}^n p_iP_i]=\sum_{i=1}^n p_itr[P_i]=\sum_{i=1}^n p_i
\end{eqnarray*}

Hence, $D\in\mathbf{S}(\mathbb{P})$ and therefore our claim at the beginning of the proof is true.\\

Now $\mathbf{S}(\mathbb{P})$ is the set of all convex combinations of the elements in $\mathbb{P}$.\\

We know that $\mathbf{S}(\mathbb{P})=\mathbb{P}$ if and only if $\mathbb{P}$ is convex.
Combining this with our previous claim we can conclude that any arbitrary state in $\mathbf{K}(\mathbb{P})$ is in $\mathbb{P}$ if and only if $\mathbb{P}$ is convex.
\end{proof}

We shall conclude our discussion of this section with the following theorem:

\begin{theorem}\label{th6.3}
Let $\mathcal{P}$ be a property on quantum channels and let $\mathbb{P}$ be the set of quantum channels that satisfy $\mathcal{P}$.
If $\mathbb{P}$ is convex, then the problem of interpolation by elements in $\mathbb{P}$ can be reconstructed as the following semi-SDP:

\begin{equation}\label{eq6.1}
\Delta =\inf\{ w\lambda+\sum_i tr[P_i+Q_i]\}
\end{equation}
Subject to
\begin{gather*}
    P_i,Q_i\geq 0\\
     P_i-Q_i =tr_2[C( I \otimes X_i^T)]-Y_i \quad \forall \; i\\
     C\in\mathbf{K}(\mathbb{P})\\
      -\lambda I \leq tr_1[C]- I \leq\lambda  I.
\end{gather*}

\end{theorem}

\section{An Explicit construction of Interpolating Entanglement Breaking Map When The Given Initial States Are Orthogonal}

In this section, we are going to solve the interpolation problem with an entanglement breaking map for a very restricted case, i.e. when the set of initial states $\{ A_i\}_{i=1}^k$ is orthogonal.\\

Moreover, we are going to prove that the above-mentioned entanglement breaking map is an EBT if $ I \in span(\{ A_i\}_{i=1}^k)$.\\

Our main result for this section is:\\
\begin{prop}\label{7}
    Let $\mathcal{A}=\{ A_1,A_2,...,A_k\}$ be an orthogonal set of positive matrices with the identity matrix $I$ in $span(\mathcal{A})$ and $\mathcal{B}=\{ B_1,B_2,...,B_k\}$ be another set of positive matrices. Then the following statements are equivalent:

(a) $tr[A_i]=tr[B_i]$ for each $i\in\{1,2,...,k\}$.

(b) There is an entanglement breaking trace-preserving map $\phi:\mathbb{M}_n\rightarrow\mathbb{M}_n$ such that $\phi(A_i)=B_i$ for each $i\in\{1,2,...,k\}$.
\end{prop}

To prove this proposition we have to prove a few lemmas first.
Those are:
\begin{lemma}\label{lm3.2}
Let $V$ be a one-dimensional subspace of $\mathbb{M}_n$ spanned by a positive matrix $A$ with $\langle A,A\rangle=1$.
Let $P_V$ be the projection map onto $V$ and let $\phi$ be a completely positive map on $\mathbb{M}_n$.
Then $\phi_V=\phi\circ P_V$ is completely positive.
\end{lemma}

\begin{proof}
By definition of projection map, $P_V(B)=\langle B,A\rangle A=tr[B^*A]A$ for any $B\in\mathbb{M}_n$.
Let assume, $A=\begin{pmatrix}
a_{ij}
\end{pmatrix}_{ij}$

Then, 
\begin{eqnarray*}
  P_V(E_{rs})=tr[E_{rs}^*A]A=tr[E_{sr}A]A=a_{sr}A  
\end{eqnarray*}

Then 
\begin{eqnarray*}
     I \otimes\phi_{V}\begin{pmatrix}
(\begin{pmatrix}
E_{ij}
\end{pmatrix}
\end{pmatrix}_{ij}=\begin{pmatrix}
\begin{pmatrix}
\phi_V(E_{ij})
\end{pmatrix}
\end{pmatrix}_{ij}=\begin{pmatrix}
\begin{pmatrix}
\phi(a_{ji}A)
\end{pmatrix}
\end{pmatrix}_{ij}
&=&\begin{pmatrix}
\begin{pmatrix}
a_{ji}\phi(A)
\end{pmatrix}
\end{pmatrix}_{ij}\\
&=&\begin{pmatrix}
a_{ji}
\end{pmatrix}_{ij}\otimes\phi(A)\\
&=&A^T\otimes\phi(A)
\end{eqnarray*}
 is positive.

Hence $\phi_V$ is completely positive.
\end{proof}

\begin{lemma}\label{lm3.3}
Let $\mathcal{A}=\{ A_1,A_2,...,A_k\}$ be an orthogonal set of positive matrices with the $I\in span(\mathcal{A})$ and $\mathcal{B}=\{ B_1,B_2,...,B_k\}$ be another set of positive matrices. Then the following statements are equivalent:

(a) $tr[A_i]=tr[B_i]$ for each $i\in\{1,2,...,k\}$.

(b) There is a completely positive trace preserving(CPT) map $\phi:\mathbb{M}_n\rightarrow\mathbb{M}_n$ such that $\phi(A_i)=B_i$ for each $i\in\{1,2,...,k\}$.
\end{lemma}
\begin{proof}
Let us assume (a).
By theorem \ref{th2.1}, there are CPT maps $\phi_i:\mathbb{M}_n\rightarrow\mathbb{M}_n$ such that $\phi_i(A_i)=B_i$ for each $i\in\{1,2,...,k\}$.

Let us define $\Tilde{\phi}_i=\phi_i\circ P_{A_i}$ where $P_{A_i}$ is the projection map onto the subspace spanned by $A_i$.
Then by lemma \ref{lm3.2}, $\Tilde{\phi}_i$ is completely positive for every $i$.

Let's define $\phi=\sum_{i=1}^k\Tilde{\phi}_i$.

Then,
\begin{eqnarray*}
\phi(A_j)=\sum_{i=1}^k\Tilde{\phi}_i(A_j)=\sum_{i=1}^k\phi_i(P_{A_i} (A_j))&=& \sum_{i=1}^k\phi_i(\frac{\langle A_j,A_i\rangle}{\langle A_i,A_i\rangle}A_i)\\
&=&\sum_{i=1}^k\phi_i(\delta_{ij}A_i)\\
&=&\sum_{i=1}^k\delta_{ij}\phi_i(A_i)\\
&=&\phi_j(A_j)=B_j
\end{eqnarray*}

Moreover, 
\begin{eqnarray*}
    I \otimes\phi\begin{pmatrix}
\begin{pmatrix}
E_{rs}
\end{pmatrix}
\end{pmatrix}=\sum_{i=1}^k  I \otimes\Tilde{\phi}_i\begin{pmatrix}
\begin{pmatrix}
E_{rs}
\end{pmatrix}
\end{pmatrix}=\sum_{i=1}^k A_i^T\otimes\phi_i(A_i)=\sum_{i=1}^k A_i^T\otimes B_i
\end{eqnarray*}

Hence clearly $ I \otimes\phi\begin{pmatrix}
\begin{pmatrix}
E_{rs}
\end{pmatrix}
\end{pmatrix}$ is positive and therefore $\phi$ is completely positive.

First, we shall prove that $\phi$ is trace-preserving on $V = span(\mathcal{A})$.

Let $A\in V$ be arbitrary.
Then $A=\sum_{i=1}^k\alpha_iA_i$ for some $\alpha_i\in\mathbb{C}$.

Hence, 
\begin{eqnarray*}
    \phi(A)=\sum_{i=1}^k\alpha_i\phi(A_i)=\sum_{i=1}^k\alpha_iB_i.
\end{eqnarray*}

Therefore,
\begin{eqnarray*}
    tr[\phi(A)]=tr[\sum_{i=1}^k\alpha_iB_i]=\sum_{i=1}^k\alpha_itr[B_i]=\sum_{i=1}^k\alpha_itr[A_i]=tr[\sum_{i=1}^k\alpha_iA_i]=tr(A)
\end{eqnarray*}

So $\phi$ is a completely positive map sending $A_i$ to $B_i$ and trace-preserving on $V=span(\mathcal{A}).$

Let $V^\perp$ be the perpendicular space of $V$.
Let $A\in\mathbb{M}_n$ be arbitrary.

Then $A$ can be written as $A=A_1+A_2$ where $A_1\in V$ and $A_2\in V^\perp$.

Hence 
\begin{eqnarray*}
tr[\phi(A)]=tr[\phi(A_1+A_2)]=tr[\phi(A_1)]+tr[\phi(A_2)]=tr[A_1]+tr[0]=tr[A_1]
\end{eqnarray*}

Hence it is enough to show that $Tr[B]=0$ for any $B\in V^\perp$.
\begin{eqnarray*}
    I\in V \Rightarrow span(I)\subseteq V \Rightarrow V^\perp\subseteq span(I)
\end{eqnarray*}

Let $P_I$ be the projection map onto $span(I)$.

Then,
\begin{eqnarray*}
tr[P_I(B)]=tr[\frac{\langle I,V\rangle}{\langle I,I\rangle} I]=\frac{\langle I,V\rangle}{\langle I,I\rangle}tr[I]=tr[B].   
\end{eqnarray*}

Now if $B\in V^\perp\subseteq span(I)^\perp$, then $Tr[B]=Tr[P_I(B)]=Tr[0]=0$.

Now (a) is obviously implied if we assume (b).

\end{proof}

To prove Proposition \ref{7}, we only need to prove that the $\phi$ we constructed in lemma \ref{lm3.3} is entanglement breaking.
The proof is as follows:

\begin{proof}
From theorem \ref{th2.1} of \cite{MR2837768} there exist an $n\times n$ non-negative matrix $D_i=(d_{pq}^{(i)})$ such that $(b_{ij})=(a_{ij})D_i$ where $(a_{ij}),(b_{ij})$ are diagonal elements of $A_i$ and $B_i$.
Let $F_{ij}$ be the matrix whose j'th column is $\left( \sqrt{d_{1j}^{(i)}},\sqrt{d_{2j}^{(i)}},\dots,\sqrt{d_{nj}^{(i)}}\right)^T$ and rest of the elements are $0$.
Then we can write, $B_i = \sum_{j=1}^nF_{ij}^*A_iF_{ij}$.
Let define, $\psi_i:\mathbb{M}_n\rightarrow\mathbb{M}_n$ as 
\begin{eqnarray*}
    \psi_i(X)=\sum_{j=1}^nF_{ij}^*XF_{ij}
\end{eqnarray*}

Now by lemma \ref{lm3.2} $P_i = P_{A_i}$, the projection map onto the subspace spanned by $A_i$ is completely positive as $P_i=I\circ P_i$.\\
By Choi-Kraus theorem, 
\begin{eqnarray*}
    P_i(X)=\sum_{j=1}^nG_{ij}^*XG_{ij}
\end{eqnarray*}
Hence
\begin{eqnarray*}
   \tilde{\psi_i}(X)=\psi_i\circ P_i(X)=\sum_{j=1}^n\sum_{l=1}^nF_{ij}^*G_{il}^*XG_{il}F_{il}=\sum_{j=1}^n\sum_{l=1}^n(G_{il}F_{ij})^*XG_{il}F_{ij} 
\end{eqnarray*}
Now $rank(F_{ij})\leq 1$ for all $i,j$.\\
Therefore, 
\begin{eqnarray*}
    rank(F_{ij}G_{il})\leq min\{rank(F_{ij}),rank(G_{il})\}\leq 1
\end{eqnarray*}

Therefore, each $\tilde{\psi_i}$ is entanglement breaking.\\
Therefore, $\psi=\sum_{i=1}^k\tilde{\psi_i}$ is also entanglement breaking.\\
Now,
\begin{eqnarray*}
    \psi(A_j)=\sum_{i=1}^k\tilde{\psi_i}(A_j)=\sum_{i=1}^k\psi_i(P_i(A_j))=\sum_{i=1}^k\delta_{ij}\psi_i(A_j)=\psi_j(A_j)=B_j
\end{eqnarray*}
We can show that $\psi$ is trace-preserving in the same way we showed $\phi$ is trace-preserving as $\psi=\phi$ on the span of $\mathcal{A}$.

\end{proof}

If we choose $\mathcal{A}$ and $\mathcal{B}$ to be sets of states then $tr[A_i]=tr[B_j]=1$ for any $i,j$.
Hence as long as we have an orthogonal set $\mathcal{A}$ of states such that $ I \in\mathbb{A}$ we can construct an EBT $\phi$ that takes $\mathcal{A}$ to $\mathcal{B}$ for any given set of states $\mathcal{B}$.

In this section, we have found a very strong result about the existence of interpolating EBT maps between two given sets of states, but it imposes many restrictions on the set of initial states.

\section*{Acknowledgement}
Arnab Roy acknowledges financial support from the Department of Mathematics, University of Delaware. Saikat Patra acknowledges financial support from CSIR JRF.

\begin{thebibliography}{DDOSS17}

\bibitem[Alb81]{MR649773}
Uhlmann~A. Alberti, P.~M., \emph{Stochasticity and partial order},
  Mathematische Monographien [Mathematical Monographs], vol.~18, VEB Deutscher
  Verlag der Wissenschaften, Berlin, 1981, Doubly stochastic maps and unitary
  mixing. \MR{649773}

\bibitem[AU80]{alberti-uhlmann-80}
P.~M. Alberti and A.~Uhlmann, \emph{A problem relating to positive linear maps
  on matrix algebras}, Rep. Math. Phys. \textbf{18} (1980), no.~2, 163--176
  (1983). \MR{730745}

\bibitem[Bha07]{bhatia1}
Rajendra Bhatia, \emph{Positive definite matrices}, Princeton Series in Applied
  Mathematics, Princeton University Press, Princeton, NJ, 2007. \MR{2284176
  (2007k:15005)}

\bibitem[Bla53]{MR56251}
David Blackwell, \emph{Equivalent comparisons of experiments}, Ann. Math.
  Statistics \textbf{24} (1953), 265--272. \MR{56251}

\bibitem[Bus12]{FB2891869}
Francesco Buscemi, \emph{Comparison of quantum statistical models: equivalent
  conditions for sufficiency}, Commun. Math. Phys \textbf{310} (2012), no.~3,
  625–647. \MR{2891869}

\bibitem[DDOSS17]{MR3671511}
Kenneth~R. Davidson, Adam Dor-On, Orr~Moshe Shalit, and Baruch Solel,
  \emph{Dilations, inclusions of matrix convex sets, and completely positive
  maps}, Int. Math. Res. Not. IMRN (2017), no.~13, 4069--4130. \MR{3671511}

\bibitem[DSR17]{David2017}
Andreas~Winter David~Sutter, Volkher B.~Scholz and Renato Renner,
  \emph{Approximate degradable quantum channels}, IEEE \textbf{63} (2017),
  no.~12, 7832 -- 7844. \MR{17364028}

\bibitem[GJB{\etalchar{+}}18]{Gour2018}
Gilad Gour, David Jennings, Francesco Buscemi, Runyao Duan, and Iman Marvian,
  \emph{Quantum majorization and a complete set of entropic conditions for
  quantum thermodynamics}, Nature Communications \textbf{9} (2018), no.~1,
  5352.

\bibitem[Hay17]{hayashi-1}
Masahito Hayashi, \emph{Quantum information theory}, second ed., Graduate Texts
  in Physics, Springer-Verlag, Berlin, 2017, Mathematical foundation.
  \MR{3558531}

\bibitem[HHH96]{H1}
Micha{\l} Horodecki, Pawe{\l} Horodecki, and Ryszard Horodecki,
  \emph{Separability of mixed states: necessary and sufficient conditions},
  Phys. Lett. A \textbf{223} (1996), no.~1-2, 1--8.

\bibitem[HJRW12]{heinosaari-12}
Teiko Heinosaari, Maria~A. Jivulescu, David Reeb, and Michael~M. Wolf,
  \emph{Extending quantum operations}, J. Math. Phys. \textbf{53} (2012),
  no.~10, 102208, 29. \MR{3050577}

\bibitem[HKT14]{hsu-jfa-14}
Ming-Hsiu Hsu, David Li-Wei Kuo, and Ming-Cheng Tsai, \emph{Completely positive
  interpolations of compact, trace-class and {S}chatten-{$p$} class operators},
  J. Funct. Anal. \textbf{267} (2014), no.~4, 1205--1240. \MR{3217062}

\bibitem[Hol98]{Holevo-coding}
A~S Holevo, \emph{Quantum coding theorems}, Russian Mathematical Surveys
  \textbf{53} (1998), no.~6, 1295--1331.

\bibitem[HSR03]{MR2001114}
Michael Horodecki, Peter~W. Shor, and Mary~Beth Ruskai, \emph{Entanglement
  breaking channels}, Rev. Math. Phys. \textbf{15} (2003), no.~6, 629--641.
  \MR{2001114 (2005d:81053)}

\bibitem[Jen10]{MR2661520}
Anna Jen\v{c}ov\'{a}, \emph{Quantum hypothesis testing and sufficient
  subalgebras}, Lett. Math. Phys. \textbf{93} (2010), no.~1, 15--27.
  \MR{2661520}

\bibitem[Jen12]{MR2997404}
Anna Jen\v{c}ov\'a, \emph{Comparison of quantum binary experiments}, Rep. Math.
  Phys. \textbf{70} (2012), no.~2, 237--249. \MR{2997404}

\bibitem[LD15]{MR3336734}
Yuan Li and Hong-Ke Du, \emph{Interpolations of entanglement breaking channels
  and equivalent conditions for completely positive maps}, J. Funct. Anal.
  \textbf{268} (2015), no.~11, 3566--3599. \MR{3336734}

\bibitem[LP11]{MR2837768}
Chi-Kwong Li and Yiu-Tung Poon, \emph{Interpolation by completely positive
  maps}, Linear Multilinear Algebra \textbf{59} (2011), no.~10, 1159--1170.
  \MR{2837768}

\bibitem[Mey93]{meyer}
Paul-Andr\'{e} Meyer, \emph{Quantum probability for probabilists}, Lecture
  Notes in Mathematics, vol. 1538, Springer-Verlag, Berlin, 1993. \MR{1222649}

\bibitem[Par06]{krp6}
K.~R. Parthasarathy, \emph{Quantum computation, quantum error correcting codes
  and information theory}, Published for the Tata Institute of Fundamental
  Research, Mumbai; by Narosa Publishing House, New Delhi, 2006. \MR{2538183}

\bibitem[Pau02]{paulsen}
Vern Paulsen, \emph{Completely bounded maps and operator algebras}, Cambridge
  Studies in Advanced Mathematics, vol.~78, Cambridge University Press,
  Cambridge, 2002. \MR{1976867 (2004c:46118)}

\bibitem[RSS80]{MR586660}
Ernst Ruch, Rudolf Schranner, and Thomas~H. Seligman, \emph{Generalization of a
  theorem by {H}ardy, {L}ittlewood, and {P}\'{o}lya}, J. Math. Anal. Appl.
  \textbf{76} (1980), no.~1, 222--229. \MR{586660}

\bibitem[Tor91]{MR1104437}
Erik Torgersen, \emph{Comparison of statistical experiments}, Encyclopedia of
  Mathematics and its Applications, vol.~36, Cambridge University Press,
  Cambridge, 1991. \MR{1104437}

\bibitem[Wat09]{JhonW2009}
John Watrous, \emph{Semidefinite programs for completely bounded norms}, Theory
  of Computing \textbf{5} (2009), no.~12, 217--238.

\end{thebibliography}
\newcommand{\etalchar}[1]{$^{#1}$}
\def\Dbar{\leavevmode\lower.6ex\hbox to 0pt{\hskip-.23ex \accent"16\hss}D}
  \def\Dbar{\leavevmode\lower.6ex\hbox to 0pt{\hskip-.23ex \accent"16\hss}D}
  \def\polhk#1{\setbox0=\hbox{#1}{\ooalign{\hidewidth
  \lower1.5ex\hbox{`}\hidewidth\crcr\unhbox0}}} \def\cprime{$'$}
\providecommand{\bysame}{\leavevmode\hbox to3em{\hrulefill}\thinspace}
\providecommand{\MR}{\relax\ifhmode\unskip\space\fi MR }
\providecommand{\MRhref}[2]{%
  \href{http://www.ams.org/mathscinet-getitem?mr=#1}{#2}
}
\providecommand{\href}[2]{#2}

\end{document}